\documentclass[11pt,letterpaper]{amsart}
\usepackage{amsmath,amssymb,amsthm}
\usepackage{enumerate}
\usepackage{color,xcolor}
\usepackage{gnuplot-lua-tikz}
\usepackage{cite, hyperref}
\usepackage[noabbrev,capitalize]{cleveref}
\DeclareMathOperator{\ADF}{ADF}
\DeclareMathOperator{\len}{len}
\newcommand{\C}{{\mathbb C}}
\newcommand{\Z}{{\mathbb Z}}
\newcommand{\conj}[1]{\overline{#1}}
\newcommand{\floor}[1]{\lfloor{#1}\rfloor}
\newcommand{\ceil}[1]{\lceil{#1}\rceil}
\newcommand{\norm}[2]{\Vert{#1}\Vert_{#2}}
\newcommand{\normp}[3]{\norm{#1}{#2}^{#3}}

\newcommand{\normtt}[1]{\normp{#1}{2}{2}}
\newcommand{\normtf}[1]{\normp{#1}{2}{4}}
\newcommand{\normff}[1]{\normp{#1}{4}{4}}
\newtheorem{theorem}{Theorem}[section]
\newtheorem{proposition}[theorem]{Proposition}
\newtheorem{lemma}[theorem]{Lemma}
\newtheorem{corollary}[theorem]{Corollary}
\theoremstyle{definition}
\newtheorem{definition}[theorem]{Definition}
\newtheorem{question}[theorem]{Question}
\title[Rudin-Shapiro-Like Sequences]{Rudin-Shapiro-Like Sequences with \\ Maximum Asymptotic Merit Factor}
\author{Daniel J.~Katz}
\author{Sangman Lee}
\author{Stanislav A.~Trunov}
\thanks{Daniel J.~Katz is with the Department of Mathematics, California State University, Northridge, United States.  Sangman Lee was with the Department of Mathematics, California State University, Northridge, United States, and is now with the Department of Mathematics, University of Colorado Boulder, United States.  Stanislav A.~Trunov was with the Department of Mathematics, California State University, Northridge, United States, and is now with the Department of Mathematics, Kansas State University, United States.  This paper is based upon work of all three authors supported in part by the National Science Foundation under Grant DMS-1500856 and upon work of Daniel J.~Katz supported in part by the National Science Foundation under Grant CCF-1815487.}
\date{14 July 2020}
\begin{document}
\begin{abstract}
Borwein and Mossinghoff investigated the Rudin-Shapiro-like sequences, which are infinite families of binary sequences, usually represented as polynomials.
Each family of Rudin-Shapiro-like sequences is obtained from a starting sequence (which we call the seed) by a recursive construction that doubles the length of the sequence at each step, and many sequences produced in this manner have exceptionally low aperiodic autocorrelation.
Borwein and Mossinghoff showed that the asymptotic autocorrelation merit factor for any such family is at most $3$, and found the seeds of length $40$ or less that produce the maximum asymptotic merit factor of $3$.
The definition of Rudin-Shapiro-like sequences was generalized by Katz, Lee, and Trunov to include sequences with arbitrary complex coefficients, among which are families of low autocorrelation polyphase sequences.
Katz, Lee, and Trunov proved that the maximum asymptotic merit factor is also $3$ for this larger class.
Here we show that a family of such Rudin-Shapiro-like sequences achieves asymptotic merit factor $3$ if and only if the seed is either of length $1$ or is the interleaving of a pair of Golay complementary sequences.
For small seed lengths where this is not possible, the optimal seeds are interleavings of pairs that are as close as possible to being complementary pairs, and  the idea of an almost-complementary pair makes sense of remarkable patterns in previously unexplained data on optimal seeds for binary Rudin-Shapiro-like sequences.
\end{abstract}
\maketitle
\section{Introduction}
This paper concerns families of Rudin-Shapiro-like sequences (usually represented as polynomials) with minimum asymptotic autocorrelation.
In this paper, we identify the polynomial $a(z)=a_0+a_1 z+\cdots+a_d z^d \in \C[z]$ of degree $d$ with the sequence $(a_0,a_1,\ldots,a_d) \in \C^{d+1}$.
The Rudin-Shapiro-like polynomials are a generalization due to Borwein and Mossinghoff \cite{Borwein-Mossinghoff} of Shapiro's polynomials \cite[p.~39]{Shapiro}.
Borwein and Mossinghoff's polynomials are examples of {\it Littlewood polynomials}, that is, polynomials with coefficients in $\{1,-1\}$, which are identified with sequences of terms from $\{1,-1\}$, that is, {\it binary sequences}.
Katz, Lee, and Trunov \cite{Katz-Lee-Trunov} showed that much of Borwein and Mossinghoff's theory has a natural generalization to polynomials with complex coefficients.
With this generalization, a family of Rudin-Shapiro-like polynomials is constructed from a starting polynomial $f_0(z) \in \C[z]$, called the {\it seed}, by applying the recursion
\begin{equation}\label{Abraham}
f_{n+1}(z) = f_n(z) + \sigma_n z^{1+\deg f_n} f_n^\dag(-z),
\end{equation}
where $\sigma_n \in \{-1,1\}$ and where the notation $\dag$ is used to indicate the conjugate reciprocal of a polynomial: if $a(z)=a_0+a_1 z+ \cdots +a_d z^d \in \C[z]$, then $a^\dag(z)=\conj{a_d} + \conj{a_{d-1}} z + \cdots + \conj{a_0} z^d$.
We require the seed $f_0$ to have nonzero constant coefficient so that $f_0^\dag$ has the same degree as $f_0$, and then it follows that
\[
1+\deg f_n=2^n(1+\deg f_0)
\]
for each $n$.
The sign $\sigma_n$ used in the $n$th step of the recursion can be chosen independently of the others, and $\sigma_0,\sigma_1,\ldots$ is called the {\it sign sequence} of the particular recursion used.
The sequence $f_0,f_1,\ldots$ of polynomials so produced is called the {\it stem} obtained from seed $f_0$ and sign sequence $\sigma_0,\sigma_1,\ldots$.
If one chooses $f_0=1$, $\sigma_0=1$ and $\sigma_n=(-1)^{n+1}$ for $n > 0$, then \cite[Theorem 1]{Brillhart-Carlitz} implies that the stem $f_0,f_1,\ldots$ one obtains is precisely Shapiro's original family of polynomials \cite[p.~39]{Shapiro}.

Since we identify the polynomial $a(z)=a_0+a_1 z+\cdots+a_d z^d \in \C[z]$ of degree $d$ with the sequence $(a_0,a_1,\ldots,a_d) \in \C^{d+1}$, we treat the two concepts interchangeably, and therefore apply terminology of sequences to polynomials and vice versa.
Thus the {\it length} of a nonzero polynomial $a(z)$, denoted $\len a$, is $1+\deg a$, and the zero polynomial has length $0$.
We use the adjectives {\it binary} and {\it Littlewood} interchangeably to indicate sequences of terms in $\{1,-1\}$, or equivalently, polynomials whose coefficients are in $\{1,-1\}$.
Shapiro's sequences are just the sequences of coefficients of Shapiro's polynomials.
Around the same time that Shapiro discovered his sequences, Golay independently produced sequences following the same recursion in his research on complementary pairs \cite{Golay-51}.
These sequences of Golay and Shapiro were independently rediscovered by Rudin \cite{Rudin}, and the associated polynomials came to be called the Rudin-Shapiro polynomials.
Their $L^4$ norm on the complex unit circle was studied by Littlewood \cite[Problem 19]{Littlewood} in connection with his investigations of flatness of polynomials.
It was realized \cite[eq.~(4.1)]{Hoholdt-Jensen} that calculating the $L^4$ norm of a polynomial on the complex unit circle is equivalent to studying the mean square magnitude of the autocorrelation of the associated sequence, a problem investigated by Golay \cite{Golay-72,Golay-75}.
Once it was recognized that the Rudin-Shapiro sequences have good autocorrelation properties, they were generalized, first by H\o holdt, Jensen, and Justesen \cite{Hoholdt-Jensen-Justesen} to allow for an arbitrary sign sequence $\sigma_0,\sigma_1,\ldots$ with seed $f_0=1$, and then by Borwein and Mossinghoff \cite{Borwein-Mossinghoff} to allow the seed to be an arbitrary Littlewood polynomial, and finally by Katz, Lee, and Trunov \cite{Katz-Lee-Trunov} to allow the seed to be any polynomial in $\C[z]$ with nonzero constant coefficient.

Sequences with low mean square autocorrelation are useful in various applications in remote sensing and communications \cite{Turyn-60,Golay-72,Golomb-Gong,Schroeder}.
If $a=(a_0,a_1,\ldots,a_{\ell-1}) \in \C^\ell$ is a sequence and $s \in \Z$, then the {\it aperiodic autocorrelation of $a$ at shift $s$} is
\[
C_{a,a}(s)=\sum_{j \in \Z} a_{j+s} \conj{a_j},
\]
where we use the convention that $a_j=0$ for all $j\not\in\{0,1,\ldots,\ell-1\}$.
One can think of comparing $a$ with a copy of itself that has been shifted $s$ places, and one makes the comparison by taking the inner product of the overlapping portions.
Note that $C_{a,a}(0)$ is just $\sum_{j \in \Z} |a_j|^2$, which is the squared Euclidean norm of the vector $a \in \C^\ell$.
In particular, if the terms of $a$ are complex numbers of unit magnitude, then $C_{a,a}(0)=\len a$.
One wants sequences $a$ where $|C_{a,a}(s)|$ is small for every nonzero $s$, while $C_{a,a}(0)$ is large: this aids in applications involving synchronization, since it implies a sharp contrast between the sequence in alignment with itself and out of alignment with itself.
To this end, we study the mean square magnitude of these values, and define the {\it autocorrelation demerit factor} of a sequence $a$ to be
\[
\ADF(a) = \frac{\sum_{s \in \Z, s\not=0} |C_{a,a}(s)|^2}{|C_{a,a}(0)|^2},
\]
which is the sum of squares of the autocorrelation values at nonzero shifts for the sequence obtained by scaling $a$ so that it has a Euclidean magnitude of $1$.
Sequences with good performance are those with small $\ADF$, since we want all the correlations at nonzero shifts to be small.
The {\it autocorrelation merit factor} of sequence $a$ is $1/\ADF(a)$, and was introduced \cite{Golay-72} and named \cite{Golay-75} by Golay.
The merit factor is more intuitive because it is large for sequences with good performance, but the demerit factor is easier to study, since it places the complicated terms in the numerator.

We now make the connection between Golay's merit factor and Littlewood's work on norms of polynomials on the complex unit circle.
We identify the sequence $(a_0,a_1,\ldots,a_{\ell-1}) \in \C^\ell$ with the polynomial $a(z)=a_0+a_1 z+ \cdots + a_{\ell-1} z^{\ell-1}$, and because we are interested in the polynomial's values on the complex unit circle, we set the convention that $\conj{a(z)}$ is the Laurent polynomial $\conj{a_0} +\conj {a_1} z^{-1} + \cdots +\conj{a_d} z^{-d}$.
 We also introduce the convention that $|a(z)|^2$ is the Laurent polynomial $a(z) \conj{a(z)}$, and then it is not hard to show that
\begin{equation}\label{Bernard}
|a(z)|^2 = \sum_{s \in \Z} C_{a,a}(s) z^s.
\end{equation}
If $a(z)$ is in the ring $\C[z,z^{-1}]$ of Laurent polynomials with complex coefficients, and if $p \geq 1$ is a real number, then the $L^p$ norm of $a$ on the complex unit circle is
\[
\norm{a}{p} = \frac{1}{2\pi} \left(\int_0^{2\pi} |a(e^{i\theta})|^p d\theta \right)^{1/p}.
\]
Then one finds \cite[\S V]{Katz} that
\[
\ADF(a) = \frac{\normff{a}}{\normtf{a}}-1.
\]
Littlewood's calculation \cite[Problem 19]{Littlewood} of the $L^4$ norm of the original family $f_0,f_1,\ldots$ of Rudin-Shapiro polynomials shows that $\lim_{n\to\infty} \ADF(f_n)=1/3$, or equivalently, that the merit factor of the polynomials approaches $3$ as their length tends to infinity.
H\o holdt, Jensen, and Justesen \cite[Theorem 2.3]{Hoholdt-Jensen-Justesen} generalized this result to any family of Rudin-Shapiro-like polynomials with seed $f_0=1$ and arbitrary sign sequence.
Borwein and Mossinghoff \cite[Theorem 1 and Corollary 1]{Borwein-Mossinghoff} made a further generalization to allow $f_0$ to be any nonzero Littlewood polynomial, in which case the asymptotic demerit factor is a function of some norms depending only upon the seed $f_0$.
The following is the further generalization by Katz, Lee, and Trunov \cite[Theorem 1.2]{Katz-Lee-Trunov} that allows $f_0$ to be any polynomial in $\C[z]$ with nonzero constant coefficient.
\begin{theorem}[Katz-Lee-Trunov, 2017]\label{Arthur}
Let $f_0 \in \C[z]$ be a polynomial with a nonzero constant coefficient and $\sigma_0,\sigma_1,\ldots$ be a sequence of elements from $\{1,-1\}$.  If $f_0,f_1,\ldots$ is the sequence of Rudin-Shapiro-like polynomials generated via recursion \eqref{Abraham}, then
\[
\lim_{n\to\infty} \ADF(f_n) = -1 + \frac{2}{3} \cdot \frac{\normff{f_0} + \normtt{f_0 \widetilde{f_0}}}{\normtf{f_0}} \geq \frac{1}{3},
\]
where $\widetilde{f_0}(z)$ denotes the polynomial $f_0(-z)$.
\end{theorem}
We say that a seed $f_0$ is {\it optimal} if the limiting autocorrelation demerit factor of its stem $f_0,f_1,\ldots$ is precisely $1/3$, as this is the lowest possible value.
Borwein and Mossinghoff \cite{Borwein-Mossinghoff} used a computer search informed by theory to determine all Littlewood polynomials $f_0$ with lengths from $1$ to $40$ that are optimal seeds.
They found that optimal Littlewood seeds exist at lengths $1$, $2$, $4$, $8$, $16$, $20$, $32$, and $40$, but no other lengths less than $40$.
More recently, Katz, Lee, and Trunov \cite{Katz-Lee-Trunov} conducted a massive distributed computer search via the Open Science Grid \cite{Pordes, Sfiligoi} to find all optimal Littlewood seeds up to length $52$, and discovered that there are also optimal Littlewood seeds of length $52$, but none with lengths from $41$ to $51$.
They also determined the lowest asymptotic demerit factor achieved by Littlewood seeds of length $\ell$ for each $\ell \in \{1,2,\ldots,52\}$; their results \cite[Table 1]{Katz-Lee-Trunov} are summarized here in \cref{Peter} and plotted in \cref{Quentin}.
\renewcommand{\arraystretch}{1.3}%
\begin{table}[!ht]
\caption{Lowest limiting autocorrelation demerit factor for seeds of lengths $1$ to $52$}\label{Peter}
\begin{center}
\begin{tabular}{|cc|cc|cc|cc|}
\hline
seed   & limiting           & seed   & limiting           & seed   & limiting            & seed   & limiting            \\
length & $\ADF$             & length & $\ADF$             & length & $\ADF$              & length & $\ADF$              \\ \hline\hline
 $1$   & $\frac{1}{3}$      & $14$   & $\frac{73}{147}$   & $27$   & $\frac{833}{2187}$  & $40$   & $\frac{1}{3}$       \\ 
 $2$   & $\frac{1}{3}$      & $15$   & $\frac{281}{675}$  & $28$   & $\frac{53}{147}$    & $41$   & $\frac{1841}{5043}$ \\  
 $3$   & $\frac{17}{27}$    & $16$   & $\frac{1}{3}$      & $29$   & $\frac{953}{2523}$  & $42$   & $\frac{521}{1323}$  \\  
 $4$   & $\frac{1}{3}$      & $17$   & $\frac{353}{867}$  & $30$   & $\frac{281}{675}$   & $43$   & $\frac{2017}{5547}$ \\  
 $5$   & $\frac{41}{75}$    & $18$   & $\frac{113}{243}$  & $31$   & $\frac{1081}{2883}$ & $44$   & $\frac{125}{363}$   \\  
 $6$   & $\frac{17}{27}$    & $19$   & $\frac{433}{1083}$ & $32$   & $\frac{1}{3}$       & $45$   & $\frac{2201}{6075}$ \\  
 $7$   & $\frac{73}{147}$   & $20$   & $\frac{1}{3}$      & $33$   & $\frac{1217}{3267}$ & $46$   & $\frac{617}{1587}$  \\  
 $8$   & $\frac{1}{3}$      & $21$   & $\frac{521}{1323}$ & $34$   & $\frac{353}{867}$   & $47$   & $\frac{2393}{6627}$ \\  
 $9$   & $\frac{113}{243}$  & $22$   & $\frac{161}{363}$  & $35$   & $\frac{1361}{3675}$ & $48$   & $\frac{73}{216}$    \\  
$10$   & $\frac{41}{75}$    & $23$   & $\frac{617}{1587}$ & $36$   & $\frac{29}{81}$     & $49$   & $\frac{2593}{7203}$ \\  
$11$   & $\frac{161}{363}$  & $24$   & $\frac{19}{54}$    & $37$   & $\frac{1513}{4107}$ & $50$   & $\frac{721}{1875}$  \\  
$12$   & $\frac{11}{27}$    & $25$   & $\frac{721}{1875}$ & $38$   & $\frac{433}{1083}$  & $51$   & $\frac{2801}{7803}$ \\  
$13$   & $\frac{217}{507}$  & $26$   & $\frac{217}{507}$  & $39$   & $\frac{1673}{4563}$ & $52$   & $\frac{1}{3}$       \\ \hline
\end{tabular}
\end{center}
\end{table}%
\renewcommand{\arraystretch}{1.0}%
\begin{figure}[!ht]
\begin{center}
\caption{Lowest limiting autocorrelation demerit factor for seeds of lengths $1$ to $52$}\label{Quentin}
\begin{tikzpicture}[gnuplot]
\path (0.000,0.000) rectangle (12.446,9.398);
\gpcolor{color=gp lt color border}
\gpsetlinetype{gp lt border}
\gpsetdashtype{gp dt solid}
\gpsetlinewidth{1.00}
\draw[gp path] (1.504,0.985)--(1.684,0.985);
\draw[gp path] (11.893,0.985)--(11.713,0.985);
\node[gp node right] at (1.320,0.985) {$0.3$};
\draw[gp path] (1.504,2.134)--(1.684,2.134);
\draw[gp path] (11.893,2.134)--(11.713,2.134);
\node[gp node right] at (1.320,2.134) {$0.35$};
\draw[gp path] (1.504,3.283)--(1.684,3.283);
\draw[gp path] (11.893,3.283)--(11.713,3.283);
\node[gp node right] at (1.320,3.283) {$0.4$};
\draw[gp path] (1.504,4.432)--(1.684,4.432);
\draw[gp path] (11.893,4.432)--(11.713,4.432);
\node[gp node right] at (1.320,4.432) {$0.45$};
\draw[gp path] (1.504,5.582)--(1.684,5.582);
\draw[gp path] (11.893,5.582)--(11.713,5.582);
\node[gp node right] at (1.320,5.582) {$0.5$};
\draw[gp path] (1.504,6.731)--(1.684,6.731);
\draw[gp path] (11.893,6.731)--(11.713,6.731);
\node[gp node right] at (1.320,6.731) {$0.55$};
\draw[gp path] (1.504,7.880)--(1.684,7.880);
\draw[gp path] (11.893,7.880)--(11.713,7.880);
\node[gp node right] at (1.320,7.880) {$0.6$};
\draw[gp path] (1.504,9.029)--(1.684,9.029);
\draw[gp path] (11.893,9.029)--(11.713,9.029);
\node[gp node right] at (1.320,9.029) {$0.65$};
\draw[gp path] (1.504,0.985)--(1.504,1.165);
\draw[gp path] (1.504,9.029)--(1.504,8.849);
\node[gp node center] at (1.504,0.677) {$0$};
\draw[gp path] (3.428,0.985)--(3.428,1.165);
\draw[gp path] (3.428,9.029)--(3.428,8.849);
\node[gp node center] at (3.428,0.677) {$10$};
\draw[gp path] (5.352,0.985)--(5.352,1.165);
\draw[gp path] (5.352,9.029)--(5.352,8.849);
\node[gp node center] at (5.352,0.677) {$20$};
\draw[gp path] (7.276,0.985)--(7.276,1.165);
\draw[gp path] (7.276,9.029)--(7.276,8.849);
\node[gp node center] at (7.276,0.677) {$30$};
\draw[gp path] (9.200,0.985)--(9.200,1.165);
\draw[gp path] (9.200,9.029)--(9.200,8.849);
\node[gp node center] at (9.200,0.677) {$40$};
\draw[gp path] (11.123,0.985)--(11.123,1.165);
\draw[gp path] (11.123,9.029)--(11.123,8.849);
\node[gp node center] at (11.123,0.677) {$50$};
\draw[gp path] (1.504,9.029)--(1.504,0.985)--(11.893,0.985)--(11.893,9.029)--cycle;
\node[gp node center,rotate=-270] at (0.246,5.007) {Minimum Asymptotic ADF};
\node[gp node center] at (6.698,0.215) {Length of Seed $f_0$};
\gpcolor{rgb color={0.000,0.000,0.000}}
\gpsetpointsize{4.00}
\gppoint{gp mark 7}{(1.696,1.750)}
\gppoint{gp mark 7}{(1.889,1.750)}
\gppoint{gp mark 7}{(2.274,1.750)}
\gppoint{gp mark 7}{(3.043,1.750)}
\gppoint{gp mark 7}{(4.582,1.750)}
\gppoint{gp mark 7}{(5.352,1.750)}
\gppoint{gp mark 7}{(7.660,1.750)}
\gppoint{gp mark 7}{(9.200,1.750)}
\gppoint{gp mark 7}{(11.508,1.750)}
\gppoint{gp mark 7}{(3.813,3.453)}
\gppoint{gp mark 7}{(6.121,2.176)}
\gppoint{gp mark 7}{(6.891,2.375)}
\gppoint{gp mark 7}{(8.430,2.318)}
\gppoint{gp mark 7}{(9.969,2.003)}
\gppoint{gp mark 7}{(10.739,1.856)}
\gppoint{gp mark 7}{(2.658,8.560)}
\gppoint{gp mark 7}{(3.428,6.653)}
\gppoint{gp mark 7}{(4.197,5.501)}
\gppoint{gp mark 7}{(4.967,4.777)}
\gppoint{gp mark 7}{(5.737,4.283)}
\gppoint{gp mark 7}{(6.506,3.927)}
\gppoint{gp mark 7}{(7.276,3.656)}
\gppoint{gp mark 7}{(8.045,3.446)}
\gppoint{gp mark 7}{(8.815,3.279)}
\gppoint{gp mark 7}{(9.584,3.141)}
\gppoint{gp mark 7}{(10.354,3.024)}
\gppoint{gp mark 7}{(11.123,2.927)}
\gppoint{gp mark 7}{(2.081,8.560)}
\gppoint{gp mark 7}{(2.466,6.653)}
\gppoint{gp mark 7}{(2.851,5.501)}
\gppoint{gp mark 7}{(3.236,4.777)}
\gppoint{gp mark 7}{(3.620,4.283)}
\gppoint{gp mark 7}{(4.005,3.927)}
\gppoint{gp mark 7}{(4.390,3.656)}
\gppoint{gp mark 7}{(4.775,3.446)}
\gppoint{gp mark 7}{(5.159,3.279)}
\gppoint{gp mark 7}{(5.544,3.141)}
\gppoint{gp mark 7}{(5.929,3.024)}
\gppoint{gp mark 7}{(6.314,2.927)}
\gppoint{gp mark 7}{(6.699,2.842)}
\gppoint{gp mark 7}{(7.083,2.771)}
\gppoint{gp mark 7}{(7.468,2.706)}
\gppoint{gp mark 7}{(7.853,2.651)}
\gppoint{gp mark 7}{(8.238,2.601)}
\gppoint{gp mark 7}{(8.622,2.555)}
\gppoint{gp mark 7}{(9.007,2.516)}
\gppoint{gp mark 7}{(9.392,2.479)}
\gppoint{gp mark 7}{(9.777,2.447)}
\gppoint{gp mark 7}{(10.162,2.417)}
\gppoint{gp mark 7}{(10.546,2.387)}
\gppoint{gp mark 7}{(10.931,2.362)}
\gppoint{gp mark 7}{(11.316,2.339)}
\gpsetlinetype{gp lt axes}
\gpsetdashtype{gp dt axes}
\gpsetlinewidth{2.00}
\draw[gp path] (1.504,1.751)--(1.609,1.751)--(1.714,1.751)--(1.819,1.751)--(1.924,1.751)%
  --(2.029,1.751)--(2.134,1.751)--(2.239,1.751)--(2.344,1.751)--(2.448,1.751)--(2.553,1.751)%
  --(2.658,1.751)--(2.763,1.751)--(2.868,1.751)--(2.973,1.751)--(3.078,1.751)--(3.183,1.751)%
  --(3.288,1.751)--(3.393,1.751)--(3.498,1.751)--(3.603,1.751)--(3.708,1.751)--(3.813,1.751)%
  --(3.918,1.751)--(4.023,1.751)--(4.127,1.751)--(4.232,1.751)--(4.337,1.751)--(4.442,1.751)%
  --(4.547,1.751)--(4.652,1.751)--(4.757,1.751)--(4.862,1.751)--(4.967,1.751)--(5.072,1.751)%
  --(5.177,1.751)--(5.282,1.751)--(5.387,1.751)--(5.492,1.751)--(5.597,1.751)--(5.702,1.751)%
  --(5.807,1.751)--(5.911,1.751)--(6.016,1.751)--(6.121,1.751)--(6.226,1.751)--(6.331,1.751)%
  --(6.436,1.751)--(6.541,1.751)--(6.646,1.751)--(6.751,1.751)--(6.856,1.751)--(6.961,1.751)%
  --(7.066,1.751)--(7.171,1.751)--(7.276,1.751)--(7.381,1.751)--(7.486,1.751)--(7.590,1.751)%
  --(7.695,1.751)--(7.800,1.751)--(7.905,1.751)--(8.010,1.751)--(8.115,1.751)--(8.220,1.751)%
  --(8.325,1.751)--(8.430,1.751)--(8.535,1.751)--(8.640,1.751)--(8.745,1.751)--(8.850,1.751)%
  --(8.955,1.751)--(9.060,1.751)--(9.165,1.751)--(9.270,1.751)--(9.374,1.751)--(9.479,1.751)%
  --(9.584,1.751)--(9.689,1.751)--(9.794,1.751)--(9.899,1.751)--(10.004,1.751)--(10.109,1.751)%
  --(10.214,1.751)--(10.319,1.751)--(10.424,1.751)--(10.529,1.751)--(10.634,1.751)--(10.739,1.751)%
  --(10.844,1.751)--(10.949,1.751)--(11.053,1.751)--(11.158,1.751)--(11.263,1.751)--(11.368,1.751)%
  --(11.473,1.751)--(11.578,1.751)--(11.683,1.751)--(11.788,1.751)--(11.893,1.751);
\gpcolor{color=gp lt color border}
\gpsetlinetype{gp lt border}
\gpsetdashtype{gp dt solid}
\gpsetlinewidth{1.00}
\draw[gp path] (1.504,9.029)--(1.504,0.985)--(11.893,0.985)--(11.893,9.029)--cycle;
\gpdefrectangularnode{gp plot 1}{\pgfpoint{1.504cm}{0.985cm}}{\pgfpoint{11.893cm}{9.029cm}}
\end{tikzpicture}
\end{center}
\end{figure}
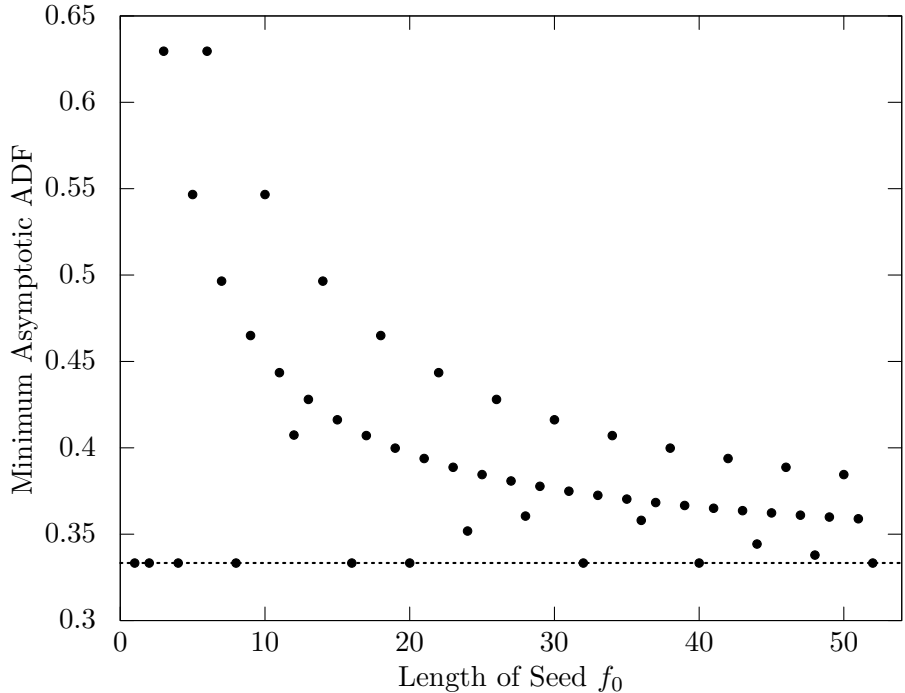%
The dotted line in \cref{Quentin} is drawn at asymptotic autocorrelation demerit factor $1/3$ to help indicate lengths for which optimal seeds exist.
We noticed a relationship between the lengths for which optimal Littlewood seeds exist and the sizes of objects known as Golay complementary pairs, which we now describe.

A {\it Golay complementary pair} (or just a {\it Golay pair} or {\it complementary pair}) is a pair of Laurent polynomials $g(z), h(z) \in \C[z,z^{-1}]$ such that $|g(z)|^2+|h(z)|^2$ is a constant.
These were first devised by Golay in \cite{Golay-51}.
If we use \eqref{Bernard} to interpret the Golay condition in terms of autocorrelation, the pair $(g,h)$ is Golay complementary if and only if $C_{g,g}(s)+C_{h,h}(s)=0$ for every nonzero shift $s$.
We note that $(g,h)$ is always a Golay complementary pair if both $g$ and $h$ are constants.
A {\it trivial} Golay complementary pair $(g,h)$ is one in which at least one of $g$ or $h$ is zero; otherwise $(g,h)$ is {\it nontrivial}.
A Golay pair $(g,h)$ is said to be {\it binary} if both $g$ and $h$ are Littlewood polynomials (that is, correspond to binary sequences).
If $(g,h)$ is a nontrivial binary Golay pair, then $g$ and $h$ must have the same length, for otherwise, the polynomial of greater length $m$ would have a nonzero correlation value at shift $m-1$, while the shorter one would have a zero correlation value at that shift, and so the sum of these correlations could not be zero.
We therefore define the {\it length} of a nontrivial binary Golay pair $(g,h)$ to be the common value $\len g=\len h$; when we speak of binary Golay pair with a length, we are asserting that it is nontrivial.
The following result due to Turyn \cite[Corollary to Lemma 5]{Turyn-74} gives all $m$ for which there are known to exist binary Golay pairs of length $m$.
\begin{theorem}[Turyn, 1974]\label{David}
For any nonnegative integers $a$, $b$, and $c$, there is a binary Golay complementary pair of length $2^a 10^b 26^c$.
\end{theorem}
A computer search by Borwein and Ferguson \cite{Borwein-Ferguson} discovered that binary Golay pairs do not exist at any length less than $100$ that is not already accounted for in this theorem.

Recall that the computer searches of Borwein-Mossinghoff and Katz-Lee-Trunov showed that optimal Littlewood seeds for the Rudin-Shapiro-like recursion exist at lengths $1$, $2$, $4$, $8$, $16$, $20$, $32$, $40$, and $52$, but no at other lengths less than $52$.
It is interesting to note that an optimal Littlewood seed of length $\ell$ with $1 < \ell \leq 52$ exists if and only if a nontrivial binary Golay pair of length $\ell/2$ exists.
There is indeed a relation between optimal seeds and Golay pairs, and to explain it we must introduce the concept of interleaving.

If $g(z), h(z) \in \C[z]$ are a pair of polynomials, then the {\it interleaving of $g$ with $h$} is $g(z^2)+z h(z^2)$.  If $g$ and $h$ both represent sequences of length $m$, then their interleaving represents the sequence $(g_0,h_0,g_1,h_1,\ldots,g_{m-1},h_{m-1})$ of length $2 m$.
Similarly, if $g$ represents a sequence of length $m+1$ and $h$ represents a sequence of length $m$, then their interleaving represents the sequence $(g_0,h_0,g_1,h_1,\ldots,g_{m-1},h_{m-1},g_m)$ of length $2 m+1$.
Now we can state our main result, which we prove later as \cref{Veronica}.
\begin{theorem}\label{Edward}
Let $f_0 \in \C[z]$ be a polynomial with a nonzero constant coefficient and $\sigma_0,\sigma_1,\ldots$ be a sequence of elements from $\{1,-1\}$.
If $f_0,f_1,\ldots$ is the sequence of Rudin-Shapiro-like polynomials generated via recursion \eqref{Abraham}, then $\lim_{n\to\infty} \ADF(f_n) \geq 1/3$, with equality if and only if $f_0$ the interleaving of a Golay complementary pair.
\end{theorem}
According to this theorem, one gets an optimal seed of length $1$ when one interleaves the trivial Golay pair $(1,0)$.
Along with \cref{David}, this tells us something about the possible lengths of optimal seeds for binary Rudin-Shapiro-like sequences.
\begin{corollary}
There exists a Littlewood polynomial $f_0$ of length $\ell$ giving rise to a sequence $(f_0,f_1,\ldots)$ of Rudin-Shapiro-like polynomials with asymptotic autocorrelation demerit factor $1/3$ if $\ell=1$ or $\ell=2^a 10^b 26^c$ for some integers $a$, $b$, $c$ with $a \geq 1$ and $b,c \geq 0$.
\end{corollary}
This explains why we find optimal Littlewood seeds at the lengths $1$, $2$, $4$, $8$, $16$, $20$, $32$, $40$, and $52$ in \cref{Peter} and \cref{Quentin}.
Borwein and Ferguson's result \cite{Borwein-Ferguson} that binary Golay pairs do not exist at any length less than $100$ that is not accounted for in \cref{David} explains why we do not see optimal Littlewood seeds at any other lengths in \cref{Peter} and \cref{Quentin}.

Although we have now settled the question of optimality, there still appears to be a lot of structure in the data, visible in \cref{Quentin}, that begs to be explained.
The points representing non-optimal seeds in \cref{Quentin} appear to lie in three families.
\begin{enumerate}[(i)]
\item\label{Eugene} The lengths that are $2$ modulo $4$ (and greater than $2$) produce a series of points that seem to be decreasing monotonically toward an asymptotic demerit factor of $1/3$ as their length increases.
\item The lengths that are odd (and greater than $1$) produce another series of points also decreasing monotonically toward an asymptotic demerit factor of $1/3$ as their length increases, and members of this series tend to be closer to $1/3$ than those of comparable length in series \eqref{Eugene}.
\item The lengths that are divisible by $4$ (and not twice the length of a binary Golay pair) tend to produce exceptionally low asymptotic demerit factors, but do not decrease monotonically.
\end{enumerate}
It turns out that most of this data can be explained by the fact that seeds of these lengths cannot be interleavings of Golay complementary pairs, but those that are closest to being optimal are interleavings of pairs $(g,h)$ of Littlewood polynomials that are very close to being complementary pairs in the sense that although $|g(z)|^2+|h(z)|^2$ is nonconstant, its $L^2$ norm is as small as certain necessary conditions on its structure allow.  It should be noted that these conditions are not sufficient to guarantee existence.
To this end, we define these almost-complementary pairs, whose structure depends on the parities of the degrees of the polynomials in the pair.
\begin{definition}\label{Yolanda}
Let $(g,h)$ be a pair of nonzero Littlewood polynomials that is not a Golay pair, and let $f(z)=\sum_{s \in \Z} f_s z^s=|g(z)|^2+|h(z)|^2$.
\begin{enumerate}[(i)]
\item\label{Eleanor} If both $g$ and $h$ are of odd length $m$, and $|f_s| \leq 2$ for every nonzero $s$, then $(g,h)$ is an {\it almost-complementary pair of odd length $m$}.
\item If one of $g$ or $h$ has length $m$ and the other has length $m+1$, and if $|f_s| \leq 1$ for every nonzero $s$, then $(g,h)$ is an {\it almost-complementary pair of uneven lengths $m$ and $m+1$}.
\item If both $g$ and $h$ are of even length $m$, and $f(z)$ has at most two nonconstant monomials, and the coefficients of these monomials have magnitude less than or equal to $2$, then $(g,h)$ is an {\it almost-complementary pair of even length $m$}.
\end{enumerate}
One may use \eqref{Bernard} to interpret the conditions in \cref{Yolanda} in terms of autocorrelation; for example, the condition in part \eqref{Eleanor} is equivalent to saying that $|C_{g,g}(s)+C_{h,h}(s)| \leq 2$ for every nonzero $s$.
Almost-complementary pairs of odd length that meet further constraints are constructed in \cite{Liu-Parampalli-Guan} and \cite{Adhikary-Majhi-Liu-Guan}.
\end{definition}
From the conditions in these definitions flow more precise consequences about the number of nonzero coefficients and their precise magnitudes, summarized Lemmas \ref{Nathan}--\ref{Sally}.
These enable us to obtain the following lower bounds on the asymptotic autocorrelation demerit factors for families of Rudin-Shapiro-like polynomials arising from binary seeds.
\begin{theorem}\label{Thomas}
Let $f_0$ be a Littlewood polynomial of length $\ell>0$, and let $g$ and $h$ be the Littlewood polynomials of lengths $\ceil{\ell/2}$ and $\floor{\ell/2}$ such that $f_0$ is the interleaving of $g$ with $h$.
Let $(f_0,f_1,\ldots)$ be a sequence of Rudin-Shapiro-like polynomials generated from seed $f_0$ with any sign sequence via recursion \eqref{Abraham}, and let $L=\lim_{n\to \infty} \ADF(f_n)$.
\begin{enumerate}[(i)]
\item\label{Abigail} If $\ell=1$, then $(g,h)$ is a trivial Golay complementary pair and $L=1/3$.
\item\label{Boris} If $\ell$ is even and there is a Golay complementary pair of length $\ell/2$, then $L \geq 1/3$, with equality if and only if $(g,h)$ is a Golay complementary pair.
\item\label{Claudia} If $\ell \equiv 0 \pmod{4}$ and there is no Golay complementary pair of length $\ell/2$, then $L \geq \frac{1}{3}+\frac{32}{3\ell^2}$, with equality if and only if $(g,h)$ is an almost-complementary pair.
\item\label{Deidre} If $\ell$ is odd and $\ell>1$, then $L \geq \frac{1}{3} + \frac{4(\ell-1)}{3\ell^2}$, with equality if and only if $(g,h)$ is an almost-complementary pair.
\item\label{Ellen} If $\ell \equiv 2 \pmod{4}$ and $\ell>2$, then there is no Golay complementary pair of length $\ell/2$.  Then $L \geq \frac{1}{3} + \frac{8(\ell-2)}{3\ell^2}$, with equality if and only if $(g,h)$ is an almost-complementary pair.
\end{enumerate}
\end{theorem}
We use \cref{Thomas} to reinterpret the data in \cref{Peter} by plotting the data again in a new \cref{Raphael}.
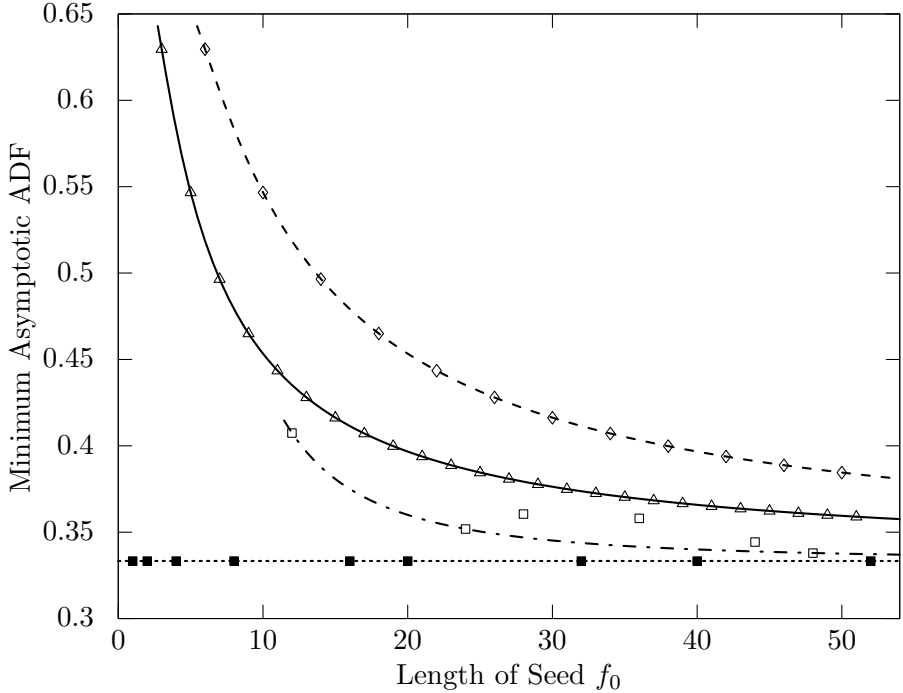
\begin{figure}[!ht]
\begin{center}
\caption{Lowest limiting autocorrelation demerit factor for seeds of lengths $1$ to $52$ with bounds from \cref{Thomas}}\label{Raphael}
\begin{tikzpicture}[gnuplot]
\path (0.000,0.000) rectangle (12.446,9.398);
\gpcolor{color=gp lt color border}
\gpsetlinetype{gp lt border}
\gpsetdashtype{gp dt solid}
\gpsetlinewidth{1.00}
\draw[gp path] (1.504,0.985)--(1.684,0.985);
\draw[gp path] (11.893,0.985)--(11.713,0.985);
\node[gp node right] at (1.320,0.985) {$0.3$};
\draw[gp path] (1.504,2.134)--(1.684,2.134);
\draw[gp path] (11.893,2.134)--(11.713,2.134);
\node[gp node right] at (1.320,2.134) {$0.35$};
\draw[gp path] (1.504,3.283)--(1.684,3.283);
\draw[gp path] (11.893,3.283)--(11.713,3.283);
\node[gp node right] at (1.320,3.283) {$0.4$};
\draw[gp path] (1.504,4.432)--(1.684,4.432);
\draw[gp path] (11.893,4.432)--(11.713,4.432);
\node[gp node right] at (1.320,4.432) {$0.45$};
\draw[gp path] (1.504,5.582)--(1.684,5.582);
\draw[gp path] (11.893,5.582)--(11.713,5.582);
\node[gp node right] at (1.320,5.582) {$0.5$};
\draw[gp path] (1.504,6.731)--(1.684,6.731);
\draw[gp path] (11.893,6.731)--(11.713,6.731);
\node[gp node right] at (1.320,6.731) {$0.55$};
\draw[gp path] (1.504,7.880)--(1.684,7.880);
\draw[gp path] (11.893,7.880)--(11.713,7.880);
\node[gp node right] at (1.320,7.880) {$0.6$};
\draw[gp path] (1.504,9.029)--(1.684,9.029);
\draw[gp path] (11.893,9.029)--(11.713,9.029);
\node[gp node right] at (1.320,9.029) {$0.65$};
\draw[gp path] (1.504,0.985)--(1.504,1.165);
\draw[gp path] (1.504,9.029)--(1.504,8.849);
\node[gp node center] at (1.504,0.677) {$0$};
\draw[gp path] (3.428,0.985)--(3.428,1.165);
\draw[gp path] (3.428,9.029)--(3.428,8.849);
\node[gp node center] at (3.428,0.677) {$10$};
\draw[gp path] (5.352,0.985)--(5.352,1.165);
\draw[gp path] (5.352,9.029)--(5.352,8.849);
\node[gp node center] at (5.352,0.677) {$20$};
\draw[gp path] (7.276,0.985)--(7.276,1.165);
\draw[gp path] (7.276,9.029)--(7.276,8.849);
\node[gp node center] at (7.276,0.677) {$30$};
\draw[gp path] (9.200,0.985)--(9.200,1.165);
\draw[gp path] (9.200,9.029)--(9.200,8.849);
\node[gp node center] at (9.200,0.677) {$40$};
\draw[gp path] (11.123,0.985)--(11.123,1.165);
\draw[gp path] (11.123,9.029)--(11.123,8.849);
\node[gp node center] at (11.123,0.677) {$50$};
\draw[gp path] (1.504,9.029)--(1.504,0.985)--(11.893,0.985)--(11.893,9.029)--cycle;
\node[gp node center,rotate=-270] at (0.246,5.007) {Minimum Asymptotic ADF};
\node[gp node center] at (6.698,0.215) {Length of Seed $f_0$};
\gpcolor{rgb color={0.000,0.000,0.000}}
\gpsetpointsize{4.00}
\gppoint{gp mark 5}{(1.696,1.750)}
\gppoint{gp mark 5}{(1.889,1.750)}
\gppoint{gp mark 5}{(2.274,1.750)}
\gppoint{gp mark 5}{(3.043,1.750)}
\gppoint{gp mark 5}{(4.582,1.750)}
\gppoint{gp mark 5}{(5.352,1.750)}
\gppoint{gp mark 5}{(7.660,1.750)}
\gppoint{gp mark 5}{(9.200,1.750)}
\gppoint{gp mark 5}{(11.508,1.750)}
\gppoint{gp mark 4}{(3.813,3.453)}
\gppoint{gp mark 4}{(6.121,2.176)}
\gppoint{gp mark 4}{(6.891,2.375)}
\gppoint{gp mark 4}{(8.430,2.318)}
\gppoint{gp mark 4}{(9.969,2.003)}
\gppoint{gp mark 4}{(10.739,1.856)}
\gpsetpointsize{6.00}
\gppoint{gp mark 12}{(2.658,8.560)}
\gppoint{gp mark 12}{(3.428,6.653)}
\gppoint{gp mark 12}{(4.197,5.501)}
\gppoint{gp mark 12}{(4.967,4.777)}
\gppoint{gp mark 12}{(5.737,4.283)}
\gppoint{gp mark 12}{(6.506,3.927)}
\gppoint{gp mark 12}{(7.276,3.656)}
\gppoint{gp mark 12}{(8.045,3.446)}
\gppoint{gp mark 12}{(8.815,3.279)}
\gppoint{gp mark 12}{(9.584,3.141)}
\gppoint{gp mark 12}{(10.354,3.024)}
\gppoint{gp mark 12}{(11.123,2.927)}
\gppoint{gp mark 8}{(2.081,8.560)}
\gppoint{gp mark 8}{(2.466,6.653)}
\gppoint{gp mark 8}{(2.851,5.501)}
\gppoint{gp mark 8}{(3.236,4.777)}
\gppoint{gp mark 8}{(3.620,4.283)}
\gppoint{gp mark 8}{(4.005,3.927)}
\gppoint{gp mark 8}{(4.390,3.656)}
\gppoint{gp mark 8}{(4.775,3.446)}
\gppoint{gp mark 8}{(5.159,3.279)}
\gppoint{gp mark 8}{(5.544,3.141)}
\gppoint{gp mark 8}{(5.929,3.024)}
\gppoint{gp mark 8}{(6.314,2.927)}
\gppoint{gp mark 8}{(6.699,2.842)}
\gppoint{gp mark 8}{(7.083,2.771)}
\gppoint{gp mark 8}{(7.468,2.706)}
\gppoint{gp mark 8}{(7.853,2.651)}
\gppoint{gp mark 8}{(8.238,2.601)}
\gppoint{gp mark 8}{(8.622,2.555)}
\gppoint{gp mark 8}{(9.007,2.516)}
\gppoint{gp mark 8}{(9.392,2.479)}
\gppoint{gp mark 8}{(9.777,2.447)}
\gppoint{gp mark 8}{(10.162,2.417)}
\gppoint{gp mark 8}{(10.546,2.387)}
\gppoint{gp mark 8}{(10.931,2.362)}
\gppoint{gp mark 8}{(11.316,2.339)}
\gpsetlinetype{gp lt axes}
\gpsetdashtype{gp dt axes}
\gpsetlinewidth{2.00}
\draw[gp path] (1.504,1.751)--(1.609,1.751)--(1.714,1.751)--(1.819,1.751)--(1.924,1.751)%
  --(2.029,1.751)--(2.134,1.751)--(2.239,1.751)--(2.344,1.751)--(2.448,1.751)--(2.553,1.751)%
  --(2.658,1.751)--(2.763,1.751)--(2.868,1.751)--(2.973,1.751)--(3.078,1.751)--(3.183,1.751)%
  --(3.288,1.751)--(3.393,1.751)--(3.498,1.751)--(3.603,1.751)--(3.708,1.751)--(3.813,1.751)%
  --(3.918,1.751)--(4.023,1.751)--(4.127,1.751)--(4.232,1.751)--(4.337,1.751)--(4.442,1.751)%
  --(4.547,1.751)--(4.652,1.751)--(4.757,1.751)--(4.862,1.751)--(4.967,1.751)--(5.072,1.751)%
  --(5.177,1.751)--(5.282,1.751)--(5.387,1.751)--(5.492,1.751)--(5.597,1.751)--(5.702,1.751)%
  --(5.807,1.751)--(5.911,1.751)--(6.016,1.751)--(6.121,1.751)--(6.226,1.751)--(6.331,1.751)%
  --(6.436,1.751)--(6.541,1.751)--(6.646,1.751)--(6.751,1.751)--(6.856,1.751)--(6.961,1.751)%
  --(7.066,1.751)--(7.171,1.751)--(7.276,1.751)--(7.381,1.751)--(7.486,1.751)--(7.590,1.751)%
  --(7.695,1.751)--(7.800,1.751)--(7.905,1.751)--(8.010,1.751)--(8.115,1.751)--(8.220,1.751)%
  --(8.325,1.751)--(8.430,1.751)--(8.535,1.751)--(8.640,1.751)--(8.745,1.751)--(8.850,1.751)%
  --(8.955,1.751)--(9.060,1.751)--(9.165,1.751)--(9.270,1.751)--(9.374,1.751)--(9.479,1.751)%
  --(9.584,1.751)--(9.689,1.751)--(9.794,1.751)--(9.899,1.751)--(10.004,1.751)--(10.109,1.751)%
  --(10.214,1.751)--(10.319,1.751)--(10.424,1.751)--(10.529,1.751)--(10.634,1.751)--(10.739,1.751)%
  --(10.844,1.751)--(10.949,1.751)--(11.053,1.751)--(11.158,1.751)--(11.263,1.751)--(11.368,1.751)%
  --(11.473,1.751)--(11.578,1.751)--(11.683,1.751)--(11.788,1.751)--(11.893,1.751);
\gpsetlinetype{gp lt border}
\gpsetdashtype{gp dt 6}
\draw[gp path] (3.708,3.620)--(3.813,3.454)--(3.918,3.309)--(4.023,3.182)--(4.127,3.069)%
  --(4.232,2.970)--(4.337,2.881)--(4.442,2.802)--(4.547,2.731)--(4.652,2.667)--(4.757,2.609)%
  --(4.862,2.556)--(4.967,2.508)--(5.072,2.464)--(5.177,2.424)--(5.282,2.387)--(5.387,2.353)%
  --(5.492,2.322)--(5.597,2.293)--(5.702,2.266)--(5.807,2.241)--(5.911,2.218)--(6.016,2.197)%
  --(6.121,2.177)--(6.226,2.158)--(6.331,2.140)--(6.436,2.124)--(6.541,2.109)--(6.646,2.094)%
  --(6.751,2.081)--(6.856,2.068)--(6.961,2.056)--(7.066,2.044)--(7.171,2.034)--(7.276,2.023)%
  --(7.381,2.014)--(7.486,2.005)--(7.590,1.996)--(7.695,1.988)--(7.800,1.980)--(7.905,1.973)%
  --(8.010,1.965)--(8.115,1.959)--(8.220,1.952)--(8.325,1.946)--(8.430,1.940)--(8.535,1.935)%
  --(8.640,1.929)--(8.745,1.924)--(8.850,1.919)--(8.955,1.915)--(9.060,1.910)--(9.165,1.906)%
  --(9.270,1.902)--(9.374,1.898)--(9.479,1.894)--(9.584,1.890)--(9.689,1.887)--(9.794,1.883)%
  --(9.899,1.880)--(10.004,1.877)--(10.109,1.874)--(10.214,1.871)--(10.319,1.868)--(10.424,1.865)%
  --(10.529,1.863)--(10.634,1.860)--(10.739,1.857)--(10.844,1.855)--(10.949,1.853)--(11.053,1.851)%
  --(11.158,1.848)--(11.263,1.846)--(11.368,1.844)--(11.473,1.842)--(11.578,1.841)--(11.683,1.839)%
  --(11.788,1.837)--(11.893,1.835);
\gpsetdashtype{gp dt 3}
\draw[gp path] (2.553,8.867)--(2.658,8.561)--(2.763,8.253)--(2.868,7.956)--(2.973,7.675)%
  --(3.078,7.411)--(3.183,7.164)--(3.288,6.935)--(3.393,6.722)--(3.498,6.524)--(3.603,6.339)%
  --(3.708,6.167)--(3.813,6.007)--(3.918,5.858)--(4.023,5.718)--(4.127,5.586)--(4.232,5.463)%
  --(4.337,5.347)--(4.442,5.238)--(4.547,5.136)--(4.652,5.039)--(4.757,4.947)--(4.862,4.860)%
  --(4.967,4.778)--(5.072,4.699)--(5.177,4.625)--(5.282,4.554)--(5.387,4.487)--(5.492,4.423)%
  --(5.597,4.361)--(5.702,4.303)--(5.807,4.247)--(5.911,4.193)--(6.016,4.141)--(6.121,4.092)%
  --(6.226,4.045)--(6.331,3.999)--(6.436,3.955)--(6.541,3.913)--(6.646,3.873)--(6.751,3.834)%
  --(6.856,3.796)--(6.961,3.760)--(7.066,3.724)--(7.171,3.691)--(7.276,3.658)--(7.381,3.626)%
  --(7.486,3.596)--(7.590,3.566)--(7.695,3.537)--(7.800,3.509)--(7.905,3.482)--(8.010,3.456)%
  --(8.115,3.431)--(8.220,3.406)--(8.325,3.382)--(8.430,3.359)--(8.535,3.336)--(8.640,3.314)%
  --(8.745,3.293)--(8.850,3.272)--(8.955,3.252)--(9.060,3.232)--(9.165,3.213)--(9.270,3.194)%
  --(9.374,3.176)--(9.479,3.158)--(9.584,3.141)--(9.689,3.124)--(9.794,3.107)--(9.899,3.091)%
  --(10.004,3.075)--(10.109,3.060)--(10.214,3.045)--(10.319,3.030)--(10.424,3.016)--(10.529,3.002)%
  --(10.634,2.988)--(10.739,2.975)--(10.844,2.962)--(10.949,2.949)--(11.053,2.936)--(11.158,2.924)%
  --(11.263,2.912)--(11.368,2.900)--(11.473,2.888)--(11.578,2.877)--(11.683,2.866)--(11.788,2.855)%
  --(11.893,2.844);
\gpsetdashtype{gp dt solid}
\draw[gp path] (2.029,8.867)--(2.134,8.253)--(2.239,7.675)--(2.344,7.164)--(2.448,6.722)%
  --(2.553,6.339)--(2.658,6.007)--(2.763,5.718)--(2.868,5.463)--(2.973,5.238)--(3.078,5.039)%
  --(3.183,4.860)--(3.288,4.699)--(3.393,4.554)--(3.498,4.423)--(3.603,4.303)--(3.708,4.193)%
  --(3.813,4.092)--(3.918,3.999)--(4.023,3.913)--(4.127,3.834)--(4.232,3.760)--(4.337,3.691)%
  --(4.442,3.626)--(4.547,3.566)--(4.652,3.509)--(4.757,3.456)--(4.862,3.406)--(4.967,3.359)%
  --(5.072,3.314)--(5.177,3.272)--(5.282,3.232)--(5.387,3.194)--(5.492,3.158)--(5.597,3.124)%
  --(5.702,3.091)--(5.807,3.060)--(5.911,3.030)--(6.016,3.002)--(6.121,2.975)--(6.226,2.949)%
  --(6.331,2.924)--(6.436,2.900)--(6.541,2.877)--(6.646,2.855)--(6.751,2.834)--(6.856,2.813)%
  --(6.961,2.793)--(7.066,2.774)--(7.171,2.756)--(7.276,2.739)--(7.381,2.721)--(7.486,2.705)%
  --(7.590,2.689)--(7.695,2.674)--(7.800,2.659)--(7.905,2.644)--(8.010,2.630)--(8.115,2.617)%
  --(8.220,2.604)--(8.325,2.591)--(8.430,2.579)--(8.535,2.567)--(8.640,2.555)--(8.745,2.544)%
  --(8.850,2.533)--(8.955,2.522)--(9.060,2.512)--(9.165,2.501)--(9.270,2.491)--(9.374,2.482)%
  --(9.479,2.472)--(9.584,2.463)--(9.689,2.454)--(9.794,2.446)--(9.899,2.437)--(10.004,2.429)%
  --(10.109,2.421)--(10.214,2.413)--(10.319,2.405)--(10.424,2.398)--(10.529,2.390)--(10.634,2.383)%
  --(10.739,2.376)--(10.844,2.369)--(10.949,2.363)--(11.053,2.356)--(11.158,2.350)--(11.263,2.343)%
  --(11.368,2.337)--(11.473,2.331)--(11.578,2.325)--(11.683,2.319)--(11.788,2.314)--(11.893,2.308);
\gpcolor{color=gp lt color border}
\gpsetlinewidth{1.00}
\draw[gp path] (1.504,9.029)--(1.504,0.985)--(11.893,0.985)--(11.893,9.029)--cycle;
\gpdefrectangularnode{gp plot 1}{\pgfpoint{1.504cm}{0.985cm}}{\pgfpoint{11.893cm}{9.029cm}}
\end{tikzpicture}
\end{center}
\end{figure}%
Any point corresponding to a length of an interleaving of a Golay pair (including length $1$ for the interleaving of a trivial Golay pair) is represented as a filled square.
All these points have limiting $\ADF$ of $1/3$, as attested by parts \eqref{Abigail} and \eqref{Boris} of the theorem, and so lie on the dotted line at $1/3$.
The remaining points lie above $1/3$.
Of these, the ones corresponding to lengths that vanish modulo $4$ are plotted as unfilled squares, and the dot-dashed curve gives the corresponding lower bound on asymptotic $\ADF$ from part \eqref{Claudia} of the theorem.
Note that the lower bound is achieved at lengths $12$, $24$, and $48$, but not at lengths $28$, $36$, and $44$.
The points corresponding to odd lengths $\ell > 1$ are plotted as triangles, and the solid curve passing through them gives the corresponding lower bound on asymptotic $\ADF$ from part \eqref{Deidre} of the theorem.
This shows that the bound is met for all relevant lengths from $3$ to $51$.
The points corresponding to lengths $\ell > 2$ with $\ell \equiv 2 \pmod{4}$ are plotted as diamonds, and the dashed curve passing through them gives the corresponding lower bound on asymptotic $\ADF$ from part \eqref{Ellen} of the theorem.
This shows that the bound is met for all relevant lengths from $6$ to $50$.
Thus we see that we actually achieve equality in the lower bounds in \cref{Thomas} for all lengths $\ell$ from $1$ to $52$ with the exceptions of $28$, $36$, and $44$.
The failures in these cases imply the nonexistence of almost-complementary pairs of lengths $14$, $18$, and $22$, which is not surprising, since the conditions for almost-complementarity at even length are much more stringent than in the other cases.

The rest of this paper is organized as follows.
\cref{Ursula} provides a proof of \cref{Edward}.
\cref{William} provides a proof of \cref{Thomas}.
\cref{Xavier} concludes with some open questions about the existence of complementary and almost-complementary pairs.

\section{Proof of \cref{Edward}}\label{Ursula}

\cref{Edward} is the corollary of two technical lemmas, which we state and prove in this section.
For the rest of this paper, we adopt the convention that if $a(z)$ is a Laurent polynomial in $\C[z,z^{-1}]$, then $\widetilde{a}(z)$ is the Laurent polynomial $a(-z)$.
We also adopt the shorthand 
\[
\int a = \frac{1}{2\pi} \int_0^{2\pi} a(e^{i\theta}) d\theta,
\]
which is just the constant coefficient of $a(z)$.

Now we can state our first technical lemma, which is on norms of interleavings.
\begin{lemma}
Suppose that $g(z),h(z) \in \C[z,z^{-1}]$ and  $f(z)=g(z^2)+z h(z^2)$.
Then
\begin{align*}
\normtt{f} & = \normtt{g} + \normtt{h}, \text{ and} \\
\normff{f} + \normtt{f \widetilde{f}} & =  2 \normtt{\,|g(z)|^2+|h(z)|^2}.
\end{align*}
\end{lemma}
\begin{proof}
The first identity is clear because $f$ is the interleaving of $g$ with $h$, so the sum of the squared magnitudes of its coefficients is equal to sum of squared magnitudes of the coefficients of both $g$ and $h$.
For the second identity, note that
\begin{align*}
\normff{f} + \normtt{f \widetilde{f}}
& = \int |f(z)|^2 \left(|g(z^2)+z h(z^2)|^2 + |g(z^2)-z h(z^2)|^2\right) \\
& = 2 \int |f(z)|^2  \left(|g(z^2)|^2 + |z h(z^2)|^2\right).
\end{align*}
Now note that $|z h(z^2)|^2=|h(z^2)|^2$ since $|z|=1$ on the complex unit circle.  We set $a(z)=|g(z)|^2+|h(z)|^2$, and so
\begin{align*}
\normff{f} + \normtt{f \widetilde{f}}
& = 2 \int |g(z^2)+z h(z^2)|^2  a(z^2)  \\
& = 2 \int \left(a(z^2) + g(z^2) z^{-1} \conj{h(z^2)} + \conj{g(z^2)} z h(z^2) \right)  a(z^2),
\end{align*}
and since $\int z^j=0$ when $j\not=0$, we can drop the terms that have odd degree to obtain
\begin{align*}
\normff{f} + \normtt{f \widetilde{f}}
& = 2 \int a(z^2)^2 \\
& = 2 \int a(z)^2 \\\
& = 2 \normtt{\,|g(z)|^2+|h(z)|^2},
\end{align*}
where we have used the fact that $\int z^{2 j}= \int z^j$ for every $j \in \Z$ in the second equality, and recall that $a(z)=|g(z)|^2+|h(z)|^2$ (which is the same as $\conj{a(z)}$) in the third.
\end{proof}
If we use this lemma in conjunction with the equality in \cref{Arthur}, we obtain a new expression for the limiting autocorrelation demerit factor.
\begin{corollary}\label{Dennis}
Let $f_0(z) \in \C[z]$ be a polynomial with a nonzero constant coefficient, and let $g(z), h(z) \in \C[z]$ be the polynomials such that $f_0$ is the interleaving of $g$ and $h$.
If $f_0,f_1,\ldots$ is the sequence of Rudin-Shapiro-like polynomials generated via recursion \eqref{Abraham} with any sign sequence, then
\[
\lim_{n\to \infty} \ADF(f_n)
= -1 + \frac{4}{3} \cdot \frac{\normtt{\,|g|^2+|h|^2}}{\left(\normtt{g}+\normtt{h}\right)^2}.
\]
\end{corollary}
Our second technical lemma bounds the ratio of norms in \cref{Dennis}.
\begin{lemma}\label{Cecilia}
If $g(z), h(z) \in \C[z]$, then
\[
\normtt{\,|g(z)|^2+|h(z)|^2} \geq \left(\normtt{g}+\normtt{h}\right)^2,
\]
with equality if and only if $(g,h)$ is a Golay complementary pair.
\end{lemma}
\begin{proof}
The inequality stems from the fact that $\normtt{\,|g(z)|^2+|h(z)|^2}$ is the sum of the squared magnitudes of the coefficients of $|g(z)|^2+|h(z)|^2$, while the real number $\normtt{g}+\normtt{h}$ is the constant coefficient of $|g(z)|^2+|h(z)|^2$.
And then it follows that equality is achieved if and only if all nonconstant coefficients of $|g(z)|^2+|h(z)|^2$ are zero, that is, if and only if $(g,h)$ is a Golay complementary pair.
\end{proof}
Applying \cref{Cecilia} to the ratio of norms in \cref{Dennis} gives the following result, which is \cref{Edward} of the Introduction.
\begin{corollary}\label{Veronica}
Let $(f_0,f_1,\ldots)$ be a sequence of Rudin-Shapiro polynomials generated from seed $f_0 \in \C[z]$.
Then $\lim_{n\to \infty} \ADF(f_n) \geq 1/3$, with equality if and only if $f_0$ the interleaving of a Golay complementary pair.
\end{corollary}

\section{Proof of \cref{Thomas}}\label{William}

In this section we characterize almost-complementary pairs and then use this characterization to prove \cref{Thomas}.
The specific conditions for almost-complementarity in \cref{Yolanda} were chosen to yield pairs $(g,h)$ such that $|g|^2+|h|^2$ is nonconstant but has an $L^2$ norm as small as certain necessary conditions on $|g|^2+|h|^2$ allow.
Though necessary, these conditions are not sufficient to guarantee existence of almost-complementary pairs.
The first set of results (\cref{Fiona}--\cref{Henriette}) set the stage for exploring these conditions via congruences modulo $4$ for coefficients of Laurent polynomials involved in correlation calculations.
\begin{lemma}\label{Fiona}
Let $f(z)$ and $g(z)$ be Littlewood polynomials of length $m$ and let $h(z)=\sum_{s \in \Z} h_s z^s = f(z) \conj{g(z)}$.  Then $h_s=0$ if $|s| \geq m$.  If $0 \leq s \leq m$, then
\[
h_s \equiv s-m+ \sum_{j=s}^{m-1} f_j + \sum_{k=0}^{m-1-s} g_k \pmod{4},
\]
and if $-m \leq s \leq 0$, then
\[
h_s \equiv -s-m+ \sum_{j=0}^{m-1+s} f_j + \sum_{k=-s}^{m-1} g_k \pmod{4}.
\]
\end{lemma}
\begin{proof}
Since the terms of $f(z)$ have degree from $0$ to $m-1$, and the terms of $\conj{g(z)}$ have degree from $-(m-1)$ to $0$, we see that $h(z)=f(z)\conj{g(z)}$ can have no monomial whose degree is lower than $-(m-1)$ or higher than $m-1$.
If $0 \leq s \leq  m$, then  $h_s = \sum_{j=0}^{m-1-s} f_{j+s} g_j$.
Note that if $u,v \in \{1,-1\}$, then $u v \equiv u+v-1 \pmod{4}$, so
\begin{align*}
h_s
& \equiv \sum_{j=0}^{m-1-s} (f_{j+s} + g_j -1) \pmod{4} \\
& = s-m + \sum_{j=s}^{m-1} f_j + \sum_{k=0}^{m-1-s} g_k.
\end{align*}
If $-m \leq s \leq 0$, then one proceeds similarly with
\begin{align*}
h_s
& = \sum_{j=-s}^{m-1} f_{j+s} g_j \\
& \equiv \sum_{j=-s}^{m-1} (f_{j+s} + g_j -1) \pmod{4} \\
& = -s-m + \sum_{j=0}^{m-1+s} f_j + \sum_{k=-s}^{m-1} g_k. \qedhere
\end{align*}
\end{proof}
Sometimes we are interested in the sum of coefficients whose indices differ by $m$; this is related to another form of correlation called {\it periodic correlation}, which differs from the aperiodic correlation considered in this paper.
\begin{lemma}
Let $f(z)$ and $g(z)$ be Littlewood polynomials of length $m$ and let $h(z)=\sum_{s \in \Z} h_s z^s = f(z) \conj{g(z)}$.  If $0 \leq s \leq m$, then
\[
h_s+h_{s-m} \equiv -m + f(1) + g(1) \pmod{4}.
\]
\end{lemma}
\begin{proof}
Since $0 \leq s \leq m$, we see that $-m \leq s-m \leq 0$, and then \cref{Fiona} gives
\begin{align*}
h_s+h_{s-m}
& \equiv s-m+ \sum_{j=s}^{m-1} f_j + \sum_{k=0}^{m-1-s} g_k -s + \sum_{j=0}^{s-1} f_j + \sum_{k=m-s}^{m-1} g_k \pmod{4} \\
& = -m + \sum_{j=0}^{m-1} f_j + \sum_{k=0}^{m-1} g_k. \qedhere
\end{align*}
\end{proof}
When one takes $g(z)=f(z)$ in the previous result, one obtains a congruence that is critical to our analysis of almost-complementary pairs.
\begin{corollary}\label{Henriette}
Let $f(z)$ be a Littlewood polynomial of length $m$ and let $h(z)=\sum_{s \in \Z} h_s z^s= |f(z)|^2$.
If $0 \leq s \leq m$, then $h_s+h_{s-m} \equiv m \pmod{4}$.
\end{corollary}
\begin{proof}
Since a Littlewood polynomial $f(z)$ of length $m$ has $f(z) \equiv 1+z+\cdots + z^{m-1} \pmod{2}$, we see that $f(1) \equiv m \pmod{2}$, and so $2 f(1) \equiv 2 m \pmod{4}$, and the previous lemma tells us that $h_s+h_{s-m} \equiv -m + 2 f(1) \pmod{4}$.
\end{proof}
Sometimes we only need congruences on coefficients modulo $2$, which are provided in the next two results, \cref{Nancy} and \cref{Gerald}.
\begin{lemma}\label{Nancy}
Let $f(z)$ and $g(z)$ be Littlewood polynomials of length $m$, let $h(z)=\sum_{s \in \Z} h_s z^s = f(z) \conj{g(z)}$.  Then $h_s=0$ if $|s| \geq m$.  If $-m \leq s \leq m$, then $h_s \equiv s+m \pmod{2}$.
\end{lemma}
\begin{proof}
When one reduces the congruences in \cref{Fiona} modulo $2$, one notes that all the summations involve $m-|s|$ terms from $\{1,-1\}$, so each summation is $m-|s|$ modulo $2$.
\end{proof}
\begin{corollary}\label{Gerald}
Let $f(z)$ be a Littlewood polynomial of length $m$ and let $h(z)=\sum_{s \in \Z} h_s z^s= |f(z)|^2$.
Then $h_s=0$ if $|s| \geq m$ and $h_0=m$.
If $|s| \leq m$, then $h_s \equiv s+m \pmod{2}$.
\end{corollary}
\begin{proof}
Apply the previous lemma with $g(z)=f(z)$, and note that $h_0$ is the sum of the squared magnitudes of the coefficients of $f$, which are $m$ elements from $\{1,-1\}$.
\end{proof}
In the next two results, we apply the above congruences to $|g|^2+|h|^2$ for pairs $(g,h)$ of Littlewood polynomials.
Pairs where $g$ and $h$ are of equal length are covered by \cref{Imogene}, while those where $\len g$ and $\len h$ differ by one are considered in \cref{James}.
\begin{lemma}\label{Imogene}
Let $g(z)$ and $h(z)$ be Littlewood polynomials of length $m$, and let $f(z)=\sum_{s \in \Z} f_s z^s = |g(z)|^2+|h(z)|^2$.  Then $f_s$ is even for every $s \in \Z$, with $f_0=2 m$ and $f_s=0$ if $|s| \geq m$.  If $0 \leq s \leq m$, then $f_s+f_{s-m} \equiv 2m \pmod{4}$.
\end{lemma}
\begin{proof}
Apply \cref{Gerald} to each of $g$ and $h$ to get the constant coefficients of $|g|^2$ and $|h|^2$, the parity of their coefficients, and the fact that their coefficients for $z^s$ vanish when $|s| \geq m$.
Apply \cref{Henriette} to each of $|g|^2$ and $|h|^2$ and sum the results to see that $f_s+f_{s-m} \equiv 2 m \pmod{4}$.
\end{proof}
\begin{lemma}\label{James}
Let $g(z)$ and $h(z)$ be Littlewood polynomials of lengths $m+1$ and $m$, respectively, and let $f(z)=\sum_{s \in \Z} f_s z^s = |g(z)|^2+|h(z)|^2$.  Then $f_s=0$ if $|s| \geq m+1$ and $f_0=2m+1$.  If $|s| \leq m$, then $f_s$ is odd.
\end{lemma}
\begin{proof}
Apply \cref{Gerald} to each of $g$ and $h$ to get the constant coefficients of $|g|^2$ and $|h|^2$, the parity of their coefficients, and the fact that their coefficients for $z^s$ vanish when $|s| \geq m+1$.
\end{proof}
The reader should now recall \cref{Yolanda}, where almost-complementary pairs of three kinds are defined.
The conditions imposed upon such pairs $(g,h)$ in the definition are not explicit enough for us to determine the $L^2$ norm of $|g|^2+|h|^2$, but the following three results, Lemmas \ref{Nathan}--\ref{Sally}, show that these conditions actually imply precise results about the number and magnitudes of nonzero coefficients in $|g|^2+|h|^2$.
\begin{lemma}\label{Nathan}
Let $(g,h)$ be a pair of Littlewood polynomials, each of odd length $m >1$, and let $f(z)=\sum_{s \in \Z} f_s z^s=|g(z)|^2+|h(z)|^2$.  Then $(g,h)$ is an almost-complementary pair if and only if $\{|f_s|,|f_{s-m}|\}=\{0,2\}$ for every $s$ with $0 < s < m$.
\end{lemma}
\begin{proof}
If $\{|f_s|,|f_{s-m}|\}=\{0,2\}$ for every $s$ with $0 < s < m$, then note that $|f_s|=0$ when $s \geq m$ by \cref{Imogene}, so $|f_s| \leq 2$ for all nonzero $s$.
Since $m > 1$ and \cref{Imogene} tells us that either $f_1$ or $f_{1-m}$ is nonzero, we see that $(g,h)$ is not a complementary pair, but is an almost-complementary pair.
Conversely, if $(g,h)$ is an almost-complementary pair, then \cref{Imogene} tells us that all the coefficients $f_s$ are even, and that $f_s+f_{s-m} \equiv 2 \pmod{4}$ when $0 < s < m$.
This means that one of $f_s$ and $f_{s-m}$ must be $0$ modulo $4$ and the other must be $2$ modulo $4$.
Since these numbers must have magnitude less than or equal to $2$, it means that one of them must be $0$ and the other must be $2$ or $-2$.
\end{proof}
\begin{lemma}\label{Oliver}
Let $(g,h)$ be a pair of Littlewood polynomials of length $m+1$ and $m>0$, respectively, and let $f(z)=\sum_{s \in \Z} f_s z^s=|g(z)|^2+|h(z)|^2$. 
Then $(g,h)$ is an almost-complementary pair if and only if $|f_s|=1$ for every $s$ with $0 < |s| \leq m$.
\end{lemma}
\begin{proof}
If $|f_s|=1$ for every $s$ with $0 < |s| \leq m$, then note that $|f_s|=0$ when $s \geq m+1$ by \cref{James}, so $|f_s| \leq 1$ for all nonzero $s$.  Since $m > 0$, and we know that $f_1$ is nonzero, we see that $(g,h)$ is not a complementary pair, but is an almost-complementary pair.
Conversely, if $(g,h)$ is an almost-complementary pair, then \cref{James} tells us that $f_s$ is odd when $0 < |s| \leq m$.
Since these numbers must have magnitude less than or equal to $1$, it means that they must have magnitude exactly $1$.
\end{proof}
\begin{lemma}\label{Sally}
Let $(g,h)$ be a pair of Littlewood polynomials, each of even length $m >2$, and let $f(z)=\sum_{j \in \Z} f_j z^j=|g(z)|^2+|h(z)|^2$.  Then $(g,h)$ is an almost-complementary pair if and only if $f_{m/2}=f_{-m/2} \in \{2,-2\}$ and $f_s=0$ for all $s\not\in\{0,m/2,-m/2\}$.
\end{lemma}
\begin{proof}
The last set of conditions given is certainly sufficient to make $(g,h)$ almost-complementary, so let us prove they are necessary.
Assume that $(g,h)$ is almost-complementary, so $0 < |f_t| \leq 2$ for some nonzero $t$.
Note that $f_{-t}=f_t$ because $\conj{f(z)}=f(z)$, so we may take $t$ to be positive, and by \cref{Imogene}, we must have $t < m$.
\cref{Imogene} also tells us the coefficients of $f$ are even, and so $f_{-t}=f_t \in \{2,-2\}$.
Furthermore, \cref{Imogene} states that $f_t+f_{t-m} \equiv 0 \pmod{4}$, so $f_{t-m}$ must be nonzero, and since $0 < t < m$, we see that $t-m\not=0$.
But $f_t$ and $f_{-t}$ are the only nonzero coefficients of $f(z)$ other than $f_0$, so $t-m$ must be equal to $-t$, and so $t=m/2$.
\end{proof}
In Propositions \ref{Lawrence}--\ref{Katherine} below, we deduce the $L^2$ norm of $|g|^2+|h|^2$ for any almost-complementary pair from the characterizations in Lemmas \ref{Nathan}--\ref{Sally}.
\begin{proposition}\label{Lawrence}
Let $g$ and $h$ be Littlewood polynomials of odd length $m>1$.
Then $(g,h)$ is not a Golay complementary pair, and
\[
\normtt{\,|g(z)|^2+|h(z)|^2} \geq (2m)^2+4(m-1),
\]
with equality if and only if $(g,h)$ is an almost-complementary pair.
\end{proposition}
\begin{proof}
Let $f(z)=\sum_{s \in \Z} f_s z^s=|g(z)|^2+|h(z)|^2$.
\cref{Imogene} tells us that $f_s=0$ when $|s| \geq m$, that $f_0=2m$, that all the coefficients $f_s$ are even, and that $f_s+f_{s-m} \equiv 2 \pmod{4}$ when $0 \leq s \leq m$.
So when $0 < s < m$, this means that one of $f_s$ and $f_{s-m}$ must be $0$ modulo $4$ and the other must be $2$ modulo $4$.
Besides proving that $(g,h)$ is not a Golay pair, this means that $|f_s|^2+|f_{s-m}|^2 \geq 0^2+2^2=4$ for $0 < s < m$, and so
\begin{align*}
\normtt{f}
& = \sum_{s \in \Z} |f_s|^2 \\
& = |f_0|^2 + \sum_{s=1}^{m-1} \left(|f_s|^2 + |f_{s-m}|^2\right) \\
& \geq (2m)^2 + 4(m-1),
\end{align*}
with equality if and only if $\{|f_s|,|f_{s-m}|\}$ is equal to $\{0,2\}$ whenever $0 < s < m$.
By \cref{Nathan}, this last condition holds if and only if $(g,h)$ is an almost-complementary pair.
\end{proof}
\begin{proposition}\label{Matilda}
Let $g$ and $h$ be Littlewood polynomials of lengths $m+1$ and $m$, respectively, with $m > 0$.
Then $(g,h)$ is not a Golay complementary pair, and
\[
\normtt{\,|g(z)|^2+|h(z)|^2} \geq (2m+1)^2+2 m,
\]
with equality if and only if $(g,h)$ is an almost-complementary pair.
\end{proposition}
\begin{proof}
Let $f(z)=\sum_{s \in \Z} f_s z^s=|g(z)|^2+|h(z)|^2$.
\cref{James} tells us that $f_s=0$ when $|s| \geq m+1$, that $f_0=2m+1$, and that $f_s$ is odd when $|s|\leq m$.
This means that $|f_s| \geq 1$ when $|s| \leq m$, so that $(g,h)$ is not a Golay pair, and
\begin{align*}
\normtt{f}
& = \sum_{s \in \Z} |f_s|^2 \\
& = |f_0|^2 + \sum_{s=1}^{m} \left(|f_{-s}|^2 + |f_s|^2\right) \\
& \geq (2m+1)^2 + 2 m,
\end{align*}
with equality if and only if $|f_s|=1$ whenever $0 < |s| \leq m$.
By \cref{Oliver}, this last condition holds if and only if $(g,h)$ is an almost-complementary pair.
\end{proof}
\begin{proposition}\label{Katherine}
Let $g$ and $h$ be Littlewood polynomials of even length $m$.
Then
\[
\normtt{\,|g(z)|^2+|h(z)|^2} \geq (2m)^2
\]
with equality if and only if $(g,h)$ is Golay complementary pair.
If $(g,h)$ is not a Golay complementary pair, then
\[
\normtt{\,|g(z)|^2+|h(z)|^2} \geq (2m)^2+8,
\]
with equality if and only if $(g,h)$ is an almost-complementary pair.
\end{proposition}
\begin{proof}
Let $f(z)=\sum_{s \in \Z} f_s z^s=|g(z)|^2+|h(z)|^2$, where we know that $f_0=2m$ by \cref{Imogene}.
If $(g,h)$ is a complementary pair, then $f(z)$ is just the constant $f_0=2m$, and so $\normtt{f}=(2m)^2$.

So henceforth let us suppose that $(g,h)$ is not a complementary pair, so $f_t\not=0$ for some nonzero $t$.
Since $\conj{f(z)}=f(z)$, we have $f_{-t}=f_t$, so we may assume that $t$ is positive, and by \cref{Imogene} we must have $t < m$.
\cref{Imogene} also says that the coefficients of $f$ are even, so $|f_t| \geq 2$.
Thus
\begin{align*}
\normtt{f}
& = \sum_{s \in \Z} |f_s|^2 \\
& \geq |f_0|^2 + |f_t|^2 + |f_{-t}|^2 \\
& \geq (2m)^2 + 2^2 + 2^2.
\end{align*}
Furthermore, $\normtt{f}=(2m)^2+8$ if and only if we have both (i) $f_t=f_{-t} \in \{2,-2\}$ and (ii) $f_s=0$ for all $s \not\in\{0,t,-t\}$.
In this case, since \cref{Imogene} makes $f_t+f_{t-m} \equiv 0 \pmod{4}$, we see that $f_{t-m}\not=0$, and since $t < m$, conditions (i) and (ii) imply that $t-m=-t$, which implies that $t=m/2$.
So $\normtt{f}=(2 m)^2+8$ if and only if $f_{m/2}=f_{-m/2} \in \{2,-2\}$ and $f_s=0$ for all $s\not \in \{0,m/2,-m/2\}$.
By \cref{Sally}, these last conditions hold if and only if $(g,h)$ is an almost-complementary pair.
\end{proof}
Now we are ready to prove \cref{Thomas} using these three propositions.
\begin{proof}[Proof of \cref{Thomas}]
In part \eqref{Abigail}, $g$ will be a nonzero constant and $h=0$, so $(g,h)$ is a trivial Golay complementary pair.
Then parts \eqref{Abigail} and \eqref{Boris} follow from \cref{Edward}.

If $\ell \equiv 0 \pmod{4}$ and there is no Golay pair of length $\ell/2$, then  since $g$ and $h$ are Littlewood polynomials of even length $m=\ell/2$, we have $\normtt{g}+\normtt{h}=\ell$, and \cref{Katherine} tells us that
\[
\normtt{\,|g(z)|^2+|h(z)|^2} \geq \ell^2+8,
\]
with equality if and only if $(g,h)$ is an almost-complementary pair.
So
\[
\frac{\normtt{\,|g|^2+|h|^2}}{\left(\normtt{f}+\normtt{g}\right)^2} \geq 1+\frac{8}{\ell^2},
\]
with equality if and only if $(g,h)$ is an almost-complementary pair.
Then we obtain part \eqref{Claudia} by applying \cref{Dennis}.

If $\ell$ is odd and greater than $1$, then $g$ and $h$ are nonzero Littlewood polynomials of lengths $m+1=(\ell+1)/2$ and $m=(\ell-1)/2$, respectively.
We have $\normtt{g}+\normtt{h}=\ell$, and \cref{Matilda} tells us that $(g,h)$ is not a Golay pair and
\[
\normtt{\,|g(z)|^2+|h(z)|^2} \geq \ell^2+\ell-1,
\]
with equality if and only if $(g,h)$ is an almost-complementary pair.
So
\[
\frac{\normtt{\,|g|^2+|h|^2}}{\left(\normtt{f}+\normtt{g}\right)^2} \geq 1+\frac{\ell-1}{\ell^2},
\]
with equality if and only if $(g,h)$ is an almost-complementary pair.
Then we obtain part \eqref{Deidre} by applying \cref{Dennis}.

If $\ell \equiv 2 \pmod{4}$ and $\ell > 2$, then $g$ and $h$ are Littlewood polynomials of odd length $m=\ell/2$, which is greater than $1$.
We have $\normtt{g}+\normtt{h}=\ell$, and \cref{Lawrence} tells us that $(g,h)$ is not a Golay pair and
\[
\normtt{\,|g(z)|^2+|h(z)|^2} \geq \ell^2+2(\ell-2),
\]
with equality if and only if $(g,h)$ is an almost-complementary pair.
So
\[
\frac{\normtt{\,|g|^2+|h|^2}}{\left(\normtt{f}+\normtt{g}\right)^2} \geq 1+\frac{2(\ell-2)}{\ell^2},
\]
with equality if and only if $(g,h)$ is an almost-complementary pair.
Then we obtain part \eqref{Ellen} by applying \cref{Dennis}.
\end{proof}

\section{Existence of complementary and almost-complementary pairs}\label{Xavier}

In the discussion following \cref{Thomas}, we noted that the asymptotic autocorrelation demerit factors listed in \cref{Peter} and plotted in \cref{Raphael} meet the bounds in \cref{Thomas} for all lengths $\ell \leq 52$ except for $28$, $36$, and $44$.
The data attest to the well-known fact that there exist binary Golay complementary pairs of lengths $1$, $2$, $4$, $8$, $10$, $16$, $20$, and $26$ (in addition to trivial Golay pairs).
The data also imply the existence of almost-complementary pairs of odd length $m$ for $3 \leq m \leq 25$, almost-complementary pairs of unequal lengths $m$ and $m+1$ for $1 \leq m \leq 25$, and almost-complementary pairs of even length $m$ for $m=6$, $12$, and $24$ (but not $m=14$, nor $18$, nor $22$).
Overall, this shows that almost-complementary pairs very often exist for polynomials of short length, with lacunae at certain even lengths that are not surprising given how much more stringent the conditions are for almost-complementary pairs of even length (cf.~\cref{Yolanda}).

This prompts the following questions about the existence of complementary and almost-complementary pairs.
The first question, about the existence of binary Golay pairs, is well-known and has been extensively studied, but remains open.
\begin{question}
For which $m$ does there exist a binary Golay complementary pair of length $m$?
\end{question}
So far the only $m$ such that there are known to exist binary Golay pairs of length $m$ are those given by \cref{David}, that is, $m$ of the form $2^a 10^b 26^c$ for $a,b,c \geq 0$.
Some necessary conditions on the length $m$ of a binary Golay pair have been found.
It is well known that if $m > 1$, then $m$ must be even (see the paper \cite[p.~84]{Golay-61} of Golay, who does not even consider pairs of length $1$ to be complementary).
In \cite[\S 3]{Eliahou-Kervaire-Saffari-90} and \cite[Lemma 1.5]{Eliahou-Kervaire-Saffari-91}, Eliahou, Kervaire, and Saffari showed that $m$ cannot be the length of a binary Golay pair if $m$ is divisible by a prime $p$ with $p\equiv 3 \pmod{4}$.
So the prime factors of $m$ can only be $2$ and primes $p$ with $p\equiv 1 \pmod{4}$.
Since each of these primes can be factored in the ring $\Z[i]$ of Gaussian integers as $(a+b i)(a-b i)$ for some $a,b \in \Z$, their product $m$ can also be factored in this way, and so $m$ must be expressible as the sum of two squares (one of which may be $0$).
The fact that $m$ is a sum of two squares was already known to Golay \cite[p.~84]{Golay-61}.
As mentioned in the Introduction, Borwein and Ferguson \cite{Borwein-Ferguson} conducted a computer search that showed that the only $m < 100$ for which binary Golay pairs of length $m$ exist are those already given by \cref{David}.

One may also ask the question about existence of almost-complementary pairs of each type enumerated in \cref{Yolanda}.
Every pair of Littlewood sequences of length $1$ is a Golay pair, so there cannot be almost-complementary pairs of length $1$, but one may ask about all other odd lengths.
\begin{question}
For which odd $m$ with $m > 1$ does there exist an almost-complementary pair of length $m$?
\end{question}
As mentioned above, our data answer this question affirmatively for all odd $m$ with $3 \leq m \leq 25$.
Furthermore, the following result indicates that almost-complementary pairs of odd length $m > 1$ will exist if $m$ is one less or one greater than the length of a Golay pair.
\begin{proposition}\label{Hubert}
Let $(g,h)$ be a binary Golay pair of length $m$, with $g(z)=\sum_{j \in \Z} g_j z^j$ and $h(z)=\sum_{j \in \Z} h_j z^j$.
\begin{enumerate}[(i)]
\item If $m > 2$, then the pair $(g(z)-g_{m-1} z^{m-1},h(z)-h_{m-1} z^{m-1})$ is an almost-complementary pair of odd length $m-1$.
\item If $m \geq 2$ and $u,v \in \{1,-1\}$, then $(g(z)+u z^m, h(z)+v z^m)$ is an almost-complementary pair of odd length $m+1$.
\end{enumerate}
\end{proposition}
\begin{proof}
Define $f(z)=|g(z)-g_{m-1} z^{m-1}|^2+|h(z)-h_{m-1} z^{m-1}|^2$ or $|g(z)+u z^m|^2+|h(z)+v z^m|^2$, depending on which case we want to prove.
By \cref{Brian} (whose statement and proof are delayed to the end of this section), we see that all nonconstant terms of $f(z)$ have coefficients of magnitude less than or equal to $2$.
Furthermore, since Golay pairs of length greater than one always have even length, the pairs we constructed from $g$ and $h$ have odd lengths greater than $1$, so they cannot be Golay pairs, and thus they are almost-complementary pairs.
\end{proof}
This proposition, along with \cref{David}, implies the existence of almost-complementary pairs of odd lengths $2^a 10^b 26^c \pm 1$ for every $a,b,c \geq 0$ with $a+b+c \geq 1$ (except not length $1$, since all pairs of length $1$ are Golay pairs).
This accounts for our observation that there exist almost-complementary pairs of lengths $3$, $5$, $7$, $9$, $11$, $15$, $17$, $19$, $21$, and $25$.
The fact that almost-complementary pairs of lengths $13$ and $23$ also exist is not accounted for by \cref{Hubert}.
We do not know if there is any odd $m > 1$ such that an almost-complementary pair of length $m$ does not exist.
Almost-complementary pairs that meet additional constraints and that have odd lengths of the form $2^a \pm 1$ and $2^a 10^b 26^c +1$ with $a > 0$ and $b,c \geq 0$ are constructed in \cite{Liu-Parampalli-Guan} and \cite{Adhikary-Majhi-Liu-Guan}.

Similarly, we may ask about existence of almost-complementary pairs of uneven length.
\begin{question}
For which $m$ with $m > 0$ does there exist an almost-complementary pair of uneven lengths $m$ and $m+1$?
\end{question}
As mentioned above, our data answer this question affirmatively for all $m$ with $1 \leq m \leq 25$.
Furthermore, the following result indicates that there will always exist an almost-complementary pair of uneven lengths when one of the two lengths is that of a Golay pair.
\begin{proposition}\label{Margaret}
Let $(g,h)$ be a binary Golay pair of length $m$, with $g(z)=\sum_{j \in \Z} g_j z^j$.
\begin{enumerate}[(i)]
\item If $m > 1$, then $(g(z)-g_{m-1} z^{m-1},h(z))$ is an almost-complementary pair of uneven lengths $m-1$ and $m$.
\item If $m > 0$ and $u \in \{1,-1\}$, then the pair $(g(z)+u z^m, h(z))$ is an almost-complementary pair of uneven lengths $m$ and $m+1$.
\end{enumerate}
\end{proposition}
\begin{proof}
Define $f(z)=|g(z)-g_{m-1} z^{m-1}|^2+|h(z)|^2$ or $|g(z)+u z^m|^2+|h(z)|^2$, depending on which case we want to prove.
By \cref{Brian} we see that all nonconstant terms of $f(z)$ have coefficients of magnitude less than or equal to $1$.
Furthermore, the pairs we constructed from $g$ and $h$ involve nonzero Littlewood polynomials of different lengths, so they cannot be Golay pairs, and thus they are almost-complementary pairs.
\end{proof}
This proposition, along with \cref{David}, implies the existence of almost-complementary pairs of uneven lengths $2^a 10^b 26^c$ and $2^a 10^b 26^c \pm 1$ for every $a,b,c \geq 0$ (except not of uneven lengths $0$ and $1$, since all pairs with lengths $0$ and $1$ are trivial Golay pairs).
This accounts for our observation that there exist almost-complementary pairs of uneven lengths $m$ and $m+1$ with $m  \in \{1, 2, 3, 4, 7, 8, 9, 10, 15, 16, 19, 20, 25\}$.
The existence of almost-complementary pairs of uneven lengths $m$ and $m+1$ for all other $m$ with $1 \leq m \leq 25$ is not accounted for by \cref{Margaret}.
We do not know if there is any $m > 0$ such that an almost-complementary pair of uneven lengths $m$ and $m+1$ does not exist.

We may also ask the question about almost-complementary pairs of even length.
\begin{question}
For which even $m$ with $m > 0$ does there exist an almost-complementary pair of length $m$?
\end{question}
A glance at \cref{Yolanda} shows that the conditions we impose for almost-complementary pairs of even length are considerably more stringent than for those of odd length or uneven length.
Our data indicate that there do exist almost-complementary pairs of lengths $6$, $12$, and $24$, but not of lengths $14$, $18$, and $22$.
For any even $m \leq 52$ such that there is a Golay pair of length $m$, our computer search for Littlewood seeds producing minimum asymptotic merit factor found a seed of length $2 m$ that is an interleaving of a Golay pair, and so did not settle the question of existence of almost-complementary pairs at these lengths.
There do exist some almost-complementary pairs of even length $m$ such that a Golay pair of length $m$ exists: for example, $(1+z,1+z)$ is an almost-complementary pair of length $2$ and $(1+z+z^2-z^3,1-z-z^2+z^3)$ is an almost-complementary pair of length $4$.

We close with the technical lemma used to prove Propositions \ref{Hubert} and \ref{Margaret}.
\begin{lemma}\label{Brian}
Let $g(z), h(z) \in \C[z]$ be polynomials whose coefficients are complex numbers of magnitude less than or equal to $1$ and suppose that $(g,h)$ is a Golay complementary pair.
Let $a(z)=u z^j$ with $u \in \C$ and $j \geq \deg g$, and let $b=v z^k$ with $v \in \C$ and $k \geq \deg h$.
Then $f(z)=\sum_{s \in \Z} f_s z^s =|g(z)+a(z)|^2+|h(z)+b(z)|^2$ has $|f_s| \leq |u|+|v|$ for all nonzero $s$.
\end{lemma}
\begin{proof}
We note that
\begin{align*}
f
& = |g|^2+|a|^2+a\conj{g}+\conj{a}g+|h|^2+|b|^2+b\conj{h}+\conj{b}h \\
& = |g|^2+|h|^2+|a|^2+|b|^2+(a \conj{g}+b\conj{h})+(\conj{a}g+\conj{b}h).
\end{align*}
Now note that $|g|^2+|h|^2$ is a constant since $(g,h)$ is a Golay pair, and that $|a|^2$ and $|b|^2$ are constants since each of $a, b$ is either zero or a monomial with a complex coefficient.
Also note that $a \conj{g}$ can only have monomials of nonnegative degree and can only have coefficients of magnitude less than or equal to $|u|$.
Similarly $b\conj{h}$ can only have monomials of nonnegative degree and can only have coefficients of magnitude less than or equal to $|v|$.
So $a\conj{g}+b\conj{h}$ contributes monomials of nonnegative degree with coefficients of magnitude less than or equal to $|u|+|v|$.
And its conjugate $\conj{a}g+\conj{b}h$ contributes monomials of nonpositive degree with coefficients of magnitude less than or equal to $|u|+|v|$.
When we sum all contributions, we see that no $f_s$ with nonzero $s$ can have magnitude greater than $|u|+|v|$.
\end{proof}

\section*{Acknowledgements}

This research was done using computing resources provided by the Open Science Grid \cite{Pordes,Sfiligoi}, which is supported by the National Science Foundation award 1148698.
The authors thank Balamurugan Desinghu and Mats Rynge, who helped them set up the calculations on the Open Science Grid.
The authors also thank an anonymous reviewer, whose comments helped improve this paper.

\end{document}